\documentclass[journal]{IEEEtran}

\usepackage{amsmath,amsfonts}
\usepackage{algorithmic}
\usepackage{array}
\usepackage[caption=false,font=normalsize,labelfont=sf,textfont=sf]{subfig}
\usepackage{textcomp}
\usepackage{stfloats}
\usepackage{url}
\usepackage{verbatim}
\usepackage{graphicx}
\usepackage{cite}

\usepackage{amssymb}
\usepackage{bm}
\usepackage{amsthm}
\usepackage{multirow}
\usepackage{tablefootnote}
\usepackage{threeparttable}
\usepackage[hidelinks]{hyperref}

\newtheorem{definition}{Definition}
\newtheorem{theorem}{Theorem}
\newtheorem{lemma}{Lemma}
\newtheorem{alg}{Algorithm}
\newtheorem{remark}{Remark}
\newtheorem{corollary}{Corollary}
\newtheorem{proposition}{Proposition}
\newtheorem{example}{Example}

\newcommand{\concept}[1]{\textbf{#1}}

\newcommand{\F}{\mathbb{F}}
\newcommand{\Fq}{\mathbb{F}_q}

\newcommand{\Fqb}{\mathbb{F}_{q^b}}
\newcommand{\Fqbv}{\mathbb{F}_{q}^b}
\newcommand{\Fqk}{\mathbb{F}_{q}^k}
\newcommand{\Fqx}{\mathbb{F}_q[x]}
\newcommand{\Fqxk}{\mathbb{F}_q[x]_{<k}}
\newcommand{\diag}{\mathrm{diag}}

\newcommand{\mC}{\mathcal{C}}
\newcommand{\mP}{\mathcal{P}}
\newcommand{\mT}{\mathcal{T}}
\newcommand{\mL}{\mathcal{L}}

\newcommand{\Ga}{\alpha}

\newcommand{\m}{\boldsymbol{m}}
\newcommand{\mpr}{\boldsymbol{m'}}
\newcommand{\bg}{\boldsymbol{g}}
\newcommand{\bc}{\boldsymbol{c}}
\newcommand{\br}{\boldsymbol{r}}
\newcommand{\be}{\boldsymbol{e}}

\newcommand{\bv}{\boldsymbol{v}}
\newcommand{\bGa}{\boldsymbol{\alpha}}

\newcommand{\rd}{\mathrm{d}}

\newcommand{\ev}{\mathrm{ev}}

\newcommand{\evav}{\ev_{\bGa, \bv}}
\newcommand{\evavp}{\ev_{\bGa, \bv'}}
\newcommand{\GRS}{\mathrm{GRS}}
\newcommand{\TGRS}{\mathrm{TGRS}}
\newcommand{\RL}{\mathrm{RL}}
\newcommand{\fail}{\texttt{fail}}

\newcommand{\enc}{\mathrm{enc}}
\newcommand{\wdeg}{\mathrm{wdeg}}

\begin{document}

\title{Efficient Decoding of Twisted GRS Codes and Roth--Lempel Codes}

\author{Runtian~Zhu
        and~Lingfei~Jin,~\IEEEmembership{Member,~IEEE}%
\thanks{Manuscript received Dec 30, 2025; revised April 13, 2026; accepted May 21, 2026. The work was supported by the National Natural Science Foundation of China under Grant 12271110. (Corresponding author: Lingfei~Jin)}%
\thanks{Runtian~Zhu and Lingfei~Jin are with the Shanghai Key Laboratory of Intelligent Information Processing, College of Computer Science and Artificial Intelligence, Fudan University, Shanghai 200433, China, and with the State Key Laboratory of Cryptology, P. O. Box, 5159, Beijing 100878, China (emails: 23110240100@m.fudan.edu.cn, lfjin@fudan.edu.cn).}%
}

\markboth{IEEE Transactions on Information Theory,~Vol.~XX, NO.~XX,~XXXX~2026}{Zhu and Jin: Efficient Decoding of Twisted GRS Codes and Roth--Lempel Codes}

\maketitle
\begin{abstract}
  MDS codes play a central role in practice due to their broad applications. To date, most known MDS codes are generalized Reed--Solomon (GRS) codes, leaving codes that are not equivalent to GRS codes comparatively less understood. Studying this non-GRS regime is therefore of intrinsic theoretical interest, and is also practically relevant since the strong algebraic structure of GRS codes can be undesirable in cryptographic settings. Among the known non-GRS codes, twisted generalized Reed--Solomon (TGRS) codes and Roth--Lempel codes are two representative families of non-GRS codes that have attracted significant attention. Though substantial work has been devoted to the construction and structural analysis of TGRS and Roth--Lempel codes, comparatively little attention has been paid to their decoding, and many problems remain open.

  In this paper, we propose list and unique decoding algorithms for TGRS codes and Roth--Lempel codes based on the Guruswami--Sudan algorithm. Under suitable parameter conditions, our algorithms achieve near-linear running time in the code length, improving upon the previously best-known quadratic-time complexity. Our TGRS decoder supports fixed-rate TGRS codes with up to $O(n^2)$ twists, substantially extending prior work that only handled the single-twist case. For Roth--Lempel codes, we provide what appears to be the first efficient decoder. Moreover, our list decoders surpass the classical unique-decoding radius for a broad range of parameters. Finally, we incorporate algebraic manipulation detection (AMD) codes into the list-decoding framework, enabling recovery of the correct message from the output list with high probability.
\end{abstract}

\begin{IEEEkeywords}
Algebraic manipulation detection codes, decoding algorithm, Guruswami--Sudan algorithm, list decoding, MDS codes, non-GRS codes, Roth--Lempel codes, twisted generalized Reed--Solomon codes.
\end{IEEEkeywords}

\section{Introduction}

\IEEEPARstart{L}{et} $q$ be a prime power and let $\Fq$ denote the finite field of order $q$. A linear $[n,k,d]_q$ code is a $k$-dimensional subspace of $\Fq^n$ with minimum (Hamming) distance $d$. One of the main tasks in coding theory is to construct codes with parameters $n,k,d$ as good as possible. The well-known Singleton bound asserts that $d \le n-k+1$. A linear code attaining this bound with equality is called a \emph{maximum distance separable} (MDS) code. MDS codes have found widespread applications in areas such as reliable communication, distributed storage systems, and secret sharing. Therefore, the constructions and properties of MDS codes have been extensively studied \cite{gulliver2008new,jin2016construction,fang2019new,shi2019new,shi2016construction,li2023several}.

To date, most of the known MDS codes belong to the class of generalized Reed--Solomon (GRS) codes, which have been extensively investigated and widely adopted owing to their well-understood algebraic structure and flexibility in parameter selection. In contrast, relatively little research has been devoted to MDS codes that are not equivalent to GRS codes. Furthermore, the strong algebraic structure of GRS codes appears to be a disadvantage for applications in code-based cryptography. Indeed, McEliece variants instantiated with GRS codes are susceptible to efficient structural attacks, most notably the Sidelnikov--Shestakov attack \cite{sidelnikov1992insecurity} and Wieschebrink's attack \cite{wieschebrink2006attack}. We refer to codes that are not equivalent to GRS codes as non-GRS codes. Since most of the known MDS codes are GRS codes, the study of non-GRS codes is of great theoretical and practical importance.

Recently, many families of non-GRS codes have been constructed. Among them, twisted generalized Reed--Solomon (TGRS) codes~\cite{beelen2017twisted,beelen2022twisted} have attracted particular interest. By introducing twists into GRS codes, TGRS codes give rise to families that are not equivalent to GRS codes. It appears that TGRS codes are promising candidates for constructing MDS and near-MDS (NMDS) codes. Consequently, extensive research has been devoted to the properties and structures of TGRS codes, including MDS and NMDS criteria~\cite{sui2022mds}, parity-check matrices~\cite{cheng2023parity}, linear complementary dual (LCD) structures~\cite{liu2021construction,huang2023mds,yang2025two}, low-dimensional hulls~\cite{wu2020twisted}, and self-orthogonal or self-dual instances~\cite{zhu2022self,liang2025+,yang2025two,huang2021mds,zhang2022class,sui2022mds,sui2023new}.

Another important family of non-GRS codes is given by the Roth--Lempel codes~\cite{roth1989construction}, which constitute the first class of codes proven to be inequivalent to GRS codes. Roth--Lempel codes are either MDS or NMDS and have been further studied in~\cite{wu2021new,han2023roth,wu2024more,liang2025extended}.

Efficient decoding algorithms are crucial for the practical deployment of error-correcting codes. It is well-known that a $[n,k]$ MDS linear code can be uniquely decoded from up to $\lfloor(n-k)/2\rfloor$ errors. To extend the error-correcting capability, list decoding was introduced in \cite{elias1957list} and \cite{wozencraft1958list}. Unlike unique decoding, list decoding allows recovery from more than $\lfloor(n-k)/2\rfloor$ errors by outputting a list of candidates.

The main reason that GRS codes have been widely applied is the availability of efficient decoding algorithms. On one hand, several well-known unique decoding algorithms have been developed, including the Berlekamp--Massey algorithm~\cite{berlekamp1966nonbinary}, the Euclidean-based algorithm~\cite{sugiyama1975method}, the Berlekamp--Welch algorithm~\cite{welch1986error}, and the method based on error-correcting pairs (ECPs)~\cite{duursma2002error}. Among these approaches, the best-known algorithms achieve near-linear time complexity in the code length~$n$. On the other hand, list decoding of GRS codes has been a hot topic\cite{guruswami2013list,guruswami2014optimal}. The celebrated Guruswami--Sudan algorithm~\cite{sudan1997decoding,guruswami1998improved} efficiently corrects up to $n - \sqrt{n(k-1)}$ errors and, with the optimization proposed by Alekhnovich~\cite{alekhnovich2002linear}, achieves near-linear running time as well.

Although decoding algorithms for GRS codes have been extensively developed, decoding TGRS codes has received comparatively less systematic study. Sun et al.~\cite{sun2024decoding} employed the Euclidean algorithm to uniquely decode two families of MDS TGRS codes and, in time $O(nq)$, corrected up to the optimal unique-decoding radius of $\lfloor (n-k)/2 \rfloor$ errors. The subsequent work~\cite{wang2025improved} extended this approach to more general families of TGRS codes and reduced the running time to $O(n^2)$. Other contributions include unique decoding algorithms based on the Berlekamp--Massey algorithm~\cite{jia2025coding}, the Berlekamp--Welch algorithm~\cite{zhang2025decoding}, and error-correcting pairs (ECPs)~\cite{he2023error}. A detailed comparison is given in Table~\ref{tab:comparison}.

  \begin{table*}
    \centering
    \caption{Comparison of existing decoding algorithms for MDS TGRS codes}
    \label{tab:comparison}

    \begin{tabular}{|c|c|c|c|c|}
    \hline
      Method & Reference & Complexity & Decoding radius & Twist number \\ \hline
      \multirow{2}{*}{Euclidean algorithm} & \cite{sun2024decoding} & $O(nq)$ & $\lfloor (n-k)/2 \rfloor$ & $1$ \\\cline{2-5}
      & \cite{wang2025improved} & $O(n^2)$ & $\lfloor (n-k)/2 \rfloor$ & $1$ \\\hline
      Berlekamp--Massey & \cite{jia2025coding} & \begin{tabular}[c]{@{}c@{}}$O(n^2q)$ for $\lfloor (n-k)/2 \rfloor$ errors,\\ $O(n^2)$ for $\lfloor (n-k-1)/2 \rfloor$ errors\end{tabular} & $\lfloor (n-k)/2 \rfloor$ & $1$ \\\hline
      Berlekamp--Welch & \cite{zhang2025decoding} & $O(n^3)$ & $\lfloor (n-k-1)/2 \rfloor$ & $1$ \\\hline
      Error-correcting pairs & \cite{he2023error} & $O(n^3)$ & $\lfloor (n-k-1)/2 \rfloor$ & $1$ \\\hline
    \end{tabular}
    
  \end{table*}

Several important problems remain open:

\begin{enumerate}
  \item The best-known unique decoding algorithms for TGRS codes have time complexity $O(n^2)$, whereas GRS codes admit near-linear-time decoding. This leaves substantial room for improving the asymptotic complexity of TGRS decoders.
  \item Existing unique decoding algorithms apply only to single-twist TGRS codes, which imposes a strong restriction. Recent cryptanalytic results on McEliece-type schemes indicate that instantiations with TGRS codes having only $O(1)$ twists are vulnerable~\cite{lavauzelle2020cryptanalysis,couvreur2025structure}. Consequently, the development of efficient decoding algorithms for TGRS codes with a super-constant number of twists is of both theoretical and practical interest.
  \item Current studies focus exclusively on unique decoding algorithms for TGRS codes and therefore correct at most $\lfloor(n-k)/2\rfloor$ errors. In particular, decoding TGRS codes beyond the classical unique-decoding radius remains unresolved.
  \item To date, no efficient decoding algorithms for Roth--Lempel codes have been reported.
\end{enumerate}

\textbf{Contribution of this work.} In this paper, we develop both list decoding and unique decoding algorithms for TGRS codes and Roth--Lempel codes based on the Guruswami--Sudan algorithm. Under certain conditions, the proposed methods yield explicit decoding algorithms for TGRS and Roth--Lempel codes with a decoding complexity of $O(n\log^2 n\log\log n)$. To the best of our knowledge, these constitute the first explicit decoding algorithms for Roth--Lempel codes.

\begin{enumerate}
  \item \textbf{Efficient decoding algorithms for Twisted GRS codes.}

  Note that any TGRS code is a subcode of a GRS code. We refer to the dimension $k'$ of the corresponding GRS code as the \emph{pseudo-dimension} of the TGRS code. Exploiting this relationship, we obtain a list decoding algorithm for TGRS codes by applying the Guruswami--Sudan decoder to the GRS code and subsequently filtering the output list via coefficient checks. Our algorithm correctly decodes whenever the number of errors is less than
  \[
    n - \sqrt{n (k'-1)},
  \]
  which exceeds the optimal unique-decoding radius whenever
  \[
    k' < \frac{(n + k)^2}{4n} + 1.
  \]

  For fixed-rate TGRS codes, and for decoding radius less than $n - \lfloor\sqrt{n(k'-1)(1+1/s)}\rfloor$ with some fixed constant $s$, the list size is shown to be constant.

  Building on this list decoder, we further derive a unique decoding algorithm for MDS TGRS codes that operates under the same condition $k' < \frac{(n + k)^2}{4n} + 1$ and achieves optimal unique-decoding radius. For fixed-rate TGRS codes, this unique decoder applies to codes with up to $\ell = O(n^2)$ twists, which is asymptotically optimal and significantly improves upon previous work that is restricted to the single-twist case $\ell = 1$.

  For fixed-rate TGRS codes with $\ell = O\bigl(n \log^2 n \log\log n\bigr)$, and decoding radius less than $n - \lfloor\sqrt{n(k'-1)(1+1/s)}\rfloor$ for some fixed constant $s$, both our list and unique decoding algorithms run in time
  \[
    O\bigl(n \log^2 n \log\log n\bigr),
  \]
  achieving near-linear complexity and improving on all previously known TGRS decoders. For simplicity, one may take $s=1$, in which case the decoding radius specializes to $n - \sqrt{2n(k'-1)}$.

  We summarize key parameters of our TGRS list and unique decoding algorithms under various assumptions in Table~\ref{tab:tgrs_parameters}.

  \begin{table*}
    \centering
    \resizebox{\textwidth}{!}{%
    \begin{threeparttable}
    \caption{Parameters of our TGRS list and unique decoding algorithms under different assumptions}
    \label{tab:tgrs_parameters}

    \begin{tabular}{|c|c|c|c|c|c|}
    \hline
      Condition & Complexity & Decoding radius & Twist number & List size & \begin{tabular}[c]{@{}c@{}}Unique-decoding\\ constraints\end{tabular} \\ \hline
      General case & $O\!\left(\Bigl(s\frac{n}{k'}\Bigr)^{O(1)}\bigl(n\log^2 n\log\log n + \ell\bigr)\right)$ & $n - \lfloor\sqrt{n(k'-1)(1+1/s)}\rfloor$ & $\bm{O(n^2)}$ (fixed rate $\frac{k}{n}$) & $\le \sqrt{\frac{ns(s+1)}{k'-1}}$ & $k' < \frac{(n+k)^2}{4n(1+1/s)} + 1$ \\ \hline
      $s = O(nk')$ \tnote{a} & polynomial time & $\bm{n - \sqrt{n(k'-1)}}$ & $\bm{O(n^2)}$ (fixed rate $\frac{k}{n}$) & $\le n^2$ & $\bm{k' < \frac{(n+k)^2}{4n} + 1}$ \\ \hline
      \begin{tabular}[c]{@{}c@{}}Fixed rate, constant $s$ and \\ $\ell = O(n\log^2 n\log\log n)$\end{tabular} & $\bm{O(n\log^2 n\log\log n)}$ & $n - \lfloor\sqrt{n(k'-1)(1+1/s)}\rfloor$ & $O(n\log^2 n\log\log n)$ & $\bm{O(1)}$ & $k' < \frac{(n+k)^2}{4n(1+1/s)} + 1$ \\ \hline
      \begin{tabular}[c]{@{}c@{}}Fixed rate, $s=1$ and \\ $\ell = O(n\log^2 n\log\log n)$\end{tabular} & $\bm{O(n\log^2 n\log\log n)}$ & $n - \lfloor\sqrt{2n(k'-1)}\rfloor$ & $O(n\log^2 n\log\log n)$ & $\bm{O(1)}$ & $k' < \frac{(n+k)^2}{8n} + 1$ \\ \hline
    \end{tabular}

    \begin{tablenotes}
    \footnotesize
    \item[a] For a more precise choice of $s$, see \cite[p.144]{guruswami2007algorithmic}
    \end{tablenotes}
    \end{threeparttable}
    }

  \end{table*}

  Furthermore, we design an AMD-assisted list decoding procedure for TGRS codes that identifies, with high probability, the uniquely transmitted message from the output list by introducing a small number of redundancy symbols to the message prior to TGRS encoding. We provide an explicit trade-off between the number of redundancy symbols and the resulting error probability, and show that an error rate of $O(1/q)$ can be achieved by adding four redundancy symbols. Relative to the list-decoding setting summarized in Table~\ref{tab:tgrs_parameters}, the AMD-assisted TGRS decoder replaces the message length $k$ by $k+2b$, where $2b$ is the number of redundancy symbols introduced by the AMD preencoding, so the relevant pseudo-dimension is that of the AMD-augmented TGRS code. Consequently, the decoding radius, complexity, admissible number of twists, and list-size bounds remain of the same form after this substitution.
  
  Finally, we compare our unique decoding algorithms for TGRS codes with previous approaches in Table~\ref{tab:comparison_ours}.

  \begin{table*}
    \centering

    \begin{threeparttable}
    \caption{Comparison of existing unique decoding algorithms for MDS TGRS codes and our algorithm}
    \label{tab:comparison_ours}

    \begin{tabular}{|c|c|c|c|c|}
    \hline
      Method & Reference & Complexity & Decoding radius & Twist number \\ \hline
      \multirow{2}{*}{Euclidean algorithm} & \cite{sun2024decoding} & $O(nq)$ & $\lfloor (n-k)/2 \rfloor$ & $1$ \\\cline{2-5}
      & \cite{wang2025improved} & $O(n^2)$ & $\lfloor (n-k)/2 \rfloor$ & $1$ \\\hline
      Berlekamp--Massey & \cite{jia2025coding} & \begin{tabular}[c]{@{}c@{}}$O(n^2q)$ for $\lfloor (n-k)/2 \rfloor$ errors,\\ $O(n^2)$ for $\lfloor (n-k-1)/2 \rfloor$ errors\end{tabular} & $\lfloor (n-k)/2 \rfloor$ & $1$ \\\hline
      Berlekamp--Welch & \cite{zhang2025decoding} & $O(n^3)$ & $\lfloor (n-k-1)/2 \rfloor$ & $1$ \\\hline
      Error-correcting pairs & \cite{he2023error} & $O(n^3)$ & $\lfloor (n-k-1)/2 \rfloor$ & $1$ \\\hline
      Guruswami--Sudan & This paper & $\bm{O(n\log^2 n \log\log n)}$ & $\bm{n - \sqrt{2n(k'-1)}}$ \tnote{a} & $\bm{O(n\log^2 n\log\log n)}$ \\\hline
    \end{tabular}
    \begin{tablenotes}
      \item[a] For a detailed analysis, see Theorem~\ref{thm:list_dec_tgrs__correctness} and Theorem~\ref{thm:unique_dec_tgrs__correctness}. Under certain conditions, this decoding radius attains the optimal unique-decoding radius of $\lfloor (n-k)/2 \rfloor$.
    \end{tablenotes}
    \end{threeparttable}
    
  \end{table*}

  \item \textbf{Efficient decoding algorithms for Roth--Lempel codes.}
  
  Puncturing a Roth--Lempel code at its last coordinate yields a GRS code. Leveraging this structural relation, we construct a list decoding algorithm for Roth--Lempel codes by first applying the Guruswami--Sudan algorithm to list decode the punctured GRS code, and then filtering the resulting list via re-encoding and Hamming-distance checks against the received word. Our algorithm correctly decodes whenever the number of errors is less than
  \[
    (n - 1) - \sqrt{(n - 1)(k - 1)}.
  \]

  This decoding radius exceeds the optimal unique-decoding radius whenever $k < n$.

  Under the same condition, we also obtain a corresponding unique decoding algorithm that achieves optimal unique-decoding radius.
  
  Finally, we extend the AMD-assisted list decoding framework developed for TGRS codes to Roth--Lempel codes. As a result, we obtain an efficient decoder for fixed-rate Roth--Lempel codes with decoding radius less than $(n-1) - \lfloor\sqrt{(n-1)(k-1)(1+1/s)}\rfloor$ for some fixed constant $s$ in time

  \[
    O\!\left(n\log^2 n\log\log n\right).
  \]

  We summarize key parameters of our Roth--Lempel list and unique decoding algorithms in Table~\ref{tab:rl_parameters}.

  \begin{table*}
    \centering
    \caption{Parameters of our Roth--Lempel list and unique decoding algorithms under different assumptions}
    \label{tab:rl_parameters}

    \resizebox{\textwidth}{!}{%
    \begin{tabular}{|c|c|c|c|c|}
    \hline
      Condition & Complexity & Decoding radius & List size & Unique-decoding constraints \\ \hline
      General case & $O\!\left(\Bigl(s\frac{n}{k}\Bigr)^{O(1)} n\log^2 n\log\log n\right)$ & $(n - 1) - \lfloor\sqrt{(n-1)(k-1)(1+1/s)}\rfloor$ & $\le \sqrt{\frac{(n-1)s(s+1)}{k-1}}$ & $(n - 1) + (k-1) - 2\sqrt{(n-1)(k-1)(1+1/s)} > 0$ \\ \hline
      $s = O(nk)$ & polynomial time & $\bm{(n - 1) - \sqrt{(n - 1)(k - 1)}}$ & $\le (n-1)^2$ & $\bm{k < n}$ \\ \hline
      Fixed rate and constant $s$ & $\bm{O(n\log^2 n\log\log n)}$ & $(n - 1) - \lfloor\sqrt{(n-1)(k-1)(1+1/s)}\rfloor$ & $\bm{O(1)}$ & $(n - 1) + (k-1) - 2\sqrt{(n-1)(k-1)(1+1/s)} > 0$ \\ \hline
    \end{tabular}
    }

  \end{table*}

  In the AMD-assisted Roth--Lempel setting, the parameters are exactly those in Table~\ref{tab:rl_parameters} after replacing $k$ with $k+2b$, where $2b$ is the number of redundancy symbols introduced by the AMD preencoding.
 
\end{enumerate}

\textbf{Organization.} In Section~\ref{sec:preliminaries}, we introduce TGRS codes and Roth--Lempel codes, recall the Guruswami--Sudan algorithm and AMD codes, and collect some auxiliary results. In Section~\ref{sec:decoding_tgrs_codes}, we present and analyze the list and unique decoders for TGRS codes. In Section~\ref{sec:decoding_tgrs_amd}, we explain how to employ AMD codes to recover the unique transmitted message with high probability. In Section~\ref{sec:decoding_rl_codes}, we describe list, unique, and AMD-assisted decoders for Roth--Lempel codes. Finally, Section~\ref{sec:conclusion} concludes the paper.

\section{Preliminaries}

\label{sec:preliminaries}

In this section, we present a brief introduction to twisted GRS codes and Roth--Lempel codes. We then introduce the Guruswami--Sudan algorithm and algebraic manipulation detection code, the main components we use in this paper.

\subsection{GRS codes, twisted GRS codes and Roth--Lempel codes}

Let $\Fq$ be a finite field with $q$ elements. Let $\bGa = (\Ga_1, \Ga_2, \ldots, \Ga_n) \in \Fq^n$ be a vector of $n$ distinct elements and $\bv = (v_1, v_2, \ldots, v_n) \in (\Fq^*)^n$ be a vector of $n$ non-zero elements. Define the evaluation map
\begin{gather*}
  \evav : \Fq[x] \longrightarrow \Fq^n, \\
  f \longmapsto \bigl(v_1\,f(\Ga_1), v_2\,f(\Ga_2), \ldots, v_n\,f(\Ga_n)\bigr).
\end{gather*}

For any set $\mP \subseteq \Fq[x]$, we write
\[
  \evav(\mP) = \{\evav(f) : f \in \mP\}.
\]

Let $k$ be an integer with $1 \le k \le n$. The $[n,k]$ \concept{generalized Reed--Solomon (GRS) code}
$\mC_{\GRS}(\bGa, \bv, k)$ is defined as
\[
  \mC_{\GRS}(\bGa, \bv, k)
    = \evav\bigl(\Fqxk\bigr),
\]
where $\Fqxk$ denotes the set of all polynomials in $\Fq[x]$ of degree
less than $k$.

\begin{definition}\cite{beelen2017twisted,beelen2022twisted}
  \label{def:tgrs}

  Let $n$, $k$ be two integers with $k < n$ and $\ell$ be a positive integer. Let $(t_i, h_i, \eta_i)$ for $i = 1, 2, \ldots, \ell$ be $\ell$ triples such that $1 \leq t_i \leq n-k$, $0 \leq h_i < k$, $\eta_i \in \Fq^*$, and $(t_i, h_i)$ are pairwise distinct. Denote the set of these $\ell$ triples as $\mT = \{(t_i, h_i, \eta_i) : 1 \leq i \leq \ell\}$. The \concept{twisted polynomial space} $\mP_{\TGRS}(k, \mT)$ is defined as:

  \[
  \begin{split}
  & \mP_{\TGRS}(k, \mT) \\
  ={} & \left\{f = \sum_{i=0}^{k-1} f_i x^i + \sum_{j=1}^{\ell} \eta_j f_{h_j} x^{k-1+t_j} : f_i \in \Fq\right\}.
  \end{split}
  \]

  Let $\bGa = (\Ga_1, \Ga_2, \ldots, \Ga_n) \in \Fq^n$ be a vector of $n$ distinct elements in $\Fq$, and let $\bv = (v_1, v_2, \ldots, v_n) \in (\Fq^*)^n$ be a vector of non-zero elements in $\Fq$. The $[n, k]$-\concept{twisted generalized Reed--Solomon (TGRS) code} $\mC_{\TGRS}(\bGa, \bv, k, \mT)$ is defined as:

  \[\mC_{\TGRS}(\bGa, \bv, k, \mT) = \evav(\mP_{\TGRS}(k, \mT)).\]

\end{definition}

Building on TGRS codes, numerous MDS and NMDS codes have been constructed\cite{beelen2022twisted,sui2022mds}, and the properties of TGRS codes have been studied extensively\cite{cheng2023parity,liu2021construction,huang2023mds,yang2025two,wu2020twisted,zhu2022self,liang2025+,huang2021mds,zhang2022class,sui2023new}.

\begin{definition}\cite{roth1989construction}
  \label{def:rl}

  Let $n$, $k$ be two integers with $k \geq 3$ and $k + 3 \leq n \leq q + 1$. Let $\bGa = (\Ga_1, \Ga_2, \ldots, \Ga_{n-1}) \in \Fq^{n-1}$ be a vector of $n-1$ distinct elements in $\Fq$, $\bv = (v_1, v_2, \ldots, v_n) \in (\Fq^*)^n$ and $\delta \in \Fq$. Define

  \[G_0 = \begin{bmatrix}
    1 & 1 & \cdots & 1 & 0 \\
    \Ga_1 & \Ga_2 & \cdots & \Ga_{n-1} & 0 \\
    \vdots & \vdots & \ddots & \vdots & \vdots \\
    \Ga_1^{k-2} & \Ga_2^{k-2} & \cdots & \Ga_{n-1}^{k-2} & 1 \\
    \Ga_1^{k-1} & \Ga_2^{k-1} & \cdots & \Ga_{n-1}^{k-1} & \delta
  \end{bmatrix}.\]

  Let $G = G_0 \cdot \diag(v_1,\ldots,v_n)$, where $\diag(v_1,\ldots,v_n)$ denotes the diagonal matrix with diagonal entries $v_1,\ldots,v_n$. The code $\mC_{\RL}(\bGa,\bv,k,\delta)$ generated by $G$ is called a $[n,k]$-\concept{Roth--Lempel code}.

\end{definition}

Roth--Lempel codes are known to be MDS or NMDS\cite{roth1989construction}, and thus exhibit strong error-correcting capability.

\subsection{Guruswami--Sudan algorithm}

The Guruswami--Sudan algorithm~\cite{sudan1997decoding,guruswami1998improved} is a list decoding algorithm for GRS codes. Given a received word $\br$ and a decoding radius $\tau$, it outputs a list $\mL$ of message polynomials whose corresponding codewords lie within Hamming distance at most $\tau$ from $\br$. The algorithm has two main stages. In the interpolation stage, it finds a bivariate polynomial $Q(x,y)$ of bounded degree that passes through all points $(\Ga_i, r_i)$. In the root-finding stage, it finds all univariate polynomials $f$ such that $Q(x,f(x)) \equiv 0$. Finally, it filters this set to keep only those polynomials whose corresponding codewords are within Hamming distance at most $\tau$ from $\br$. We summarize the Guruswami--Sudan algorithm below.

\begin{alg}
  \label{alg:guruswami_sudan}

  Guruswami--Sudan algorithm for decoding $\mC_{\GRS}(\bGa, \bv, k)$.

  Input: A received word $\br \in \Fq^n$ and a decoding radius $\tau$. Let $r'_i = r_i/v_i$.

  Output: A list $\mL$ of polynomials in $\Fq[x]_{<k}$ or \fail{}.

  \begin{enumerate}
    \item \label{step:guruswami_sudan__interpolation}
      (Interpolation) Choose a suitable multiplicity $s$ such that $\tau < n - \lfloor \sqrt{n(k-1)(1+1/s)} \rfloor$. If no such $s$ exists, output \fail{}. Find a non-zero bivariate polynomial $Q(x,y)$ such that
      \begin{enumerate}
        \item $Q(\Ga_i, r'_i) = 0$ with multiplicity at least $s$ for all $i = 1,\dots,n$;
        \item $\wdeg_{(1, k-1)}(Q) \le \lfloor \sqrt{n(k-1)s(s+1)} \rfloor$.
      \end{enumerate}

    \item \label{step:guruswami_sudan__root_finding}
      (Root-finding) Find all univariate polynomials $f(x)$ of degree less than $k$ such that $Q(x,f(x)) \equiv 0$, and check whether $\rd(\evav(f), \br) \le \tau$. If so, add $f$ to $\mL$. Finally, output $\mL$.
  \end{enumerate}
\end{alg}

In the algorithm, $\wdeg_{(1,k-1)}(Q)$ denotes the $(1,k\!-\!1)$-weighted degree of $Q$, defined by
$\wdeg_{(1,k-1)}(Q) = \max\bigl\{i+(k-1)j: x^iy^j \text{ appears in } Q \text{ with nonzero coefficient}\bigr\}$.

Over the years, the running time of this algorithm has been improved in a sequence of works. In particular, Alekhnovich~\cite{alekhnovich2002linear} gave a near-linear-time implementation. We recall Alekhnovich's result below.

\begin{theorem}\cite{alekhnovich2002linear,guruswami2007algorithmic}
  \label{thm:guruswami_sudan}
  Algorithm~\ref{alg:guruswami_sudan} outputs a list
  \[
    \mL = \{f \in \Fq[x]_{<k} : \rd(\evav(f), \br) \le \tau\}
  \]
  in time
  \[
    O\!\left(\Bigl(s\frac{n}{k}\Bigr)^{O(1)} n\log^2 n\log\log n\right),
  \]
  provided that $\tau < n - \lfloor \sqrt{n(k-1)(1+1/s)} \rfloor$.

  In particular, Algorithm~\ref{alg:guruswami_sudan} succeeds whenever
  \[
    \tau < n - \sqrt{n(k-1)}.
  \]

\end{theorem}

List size is also a fundamental parameter in list decoding. We recall below a standard bound on the list size returned by the Guruswami--Sudan algorithm~\cite{guruswami2007algorithmic}.

\begin{lemma}\cite[Theorem 4.8]{guruswami2007algorithmic}
  \label{lem:bound_on_list_size}
  Let $\mL$ be the list output by Algorithm~\ref{alg:guruswami_sudan} when decoding the GRS code
  $\mC_{\GRS}(\bGa,\bv,k)$ of length $n$. Then
  \[
    |\mL| \le \sqrt{ns(s+1)/(k-1)},
  \]
  where $s$ is the multiplicity.
\end{lemma}

\subsection{Algebraic manipulation detection codes}

Algebraic manipulation detection (AMD) codes were originally introduced to detect algebraic tampering in linear secret sharing schemes. Subsequently, they have been used to filter out incorrect candidates in list decoding for GRS codes. We recall the relevant definitions and an optimal systematic construction.

\begin{definition}
  \label{def:amd}
  \cite{cramer2013algebraic}
  Let $S$ be a set of size $M > 1$ and let $G$ be a commutative group of order $N$. Let $E$ be a probabilistic encoding function $E : S \to G$ and $D$ a deterministic decoding function $D : G \to S \cup \{\fail{}\}$. The pair $(E, D)$ is called an \concept{$(M, N, \varepsilon)$-algebraic manipulation detection (AMD) code} if
  \begin{itemize}
    \item for all $\m \in S$, we have $D(E(\m)) = \m$;
    \item for all $\m \in S$ and all $\delta \in G \setminus \{0\}$,
    \[
      \Pr\bigl[D(E(\m) + \delta) \neq \fail{}\bigr] \le \varepsilon,
    \]
    where the probability is taken over the randomness of the encoder $E$.
  \end{itemize}
\end{definition}

\begin{definition}
  \label{def:amd_systematic}
  \cite{cramer2013algebraic}
  An $(M, N, \varepsilon)$-AMD code $(E, D)$ is called \concept{systematic} if the source set $S$ is a group and the encoding function $E$ has the form
  \[
    E : S \to G = S \times G_1 \times G_2,\quad
    s \longmapsto (s, x, g(x, s)),
  \]
  where $G_1$ and $G_2$ are groups, the element $x \in G_1$ is chosen uniformly at random, and $g : G_1 \times S \to G_2$ is a deterministic function.
\end{definition}
Denote $N_1 = |G_1|$ and $N_2 = |G_2|$, so that $N = |G| = M N_1 N_2$. The next lemma gives a general lower bound on the parameters of systematic AMD codes.

\begin{lemma}
  \cite{cramer2013algebraic}
  \label{lem:amd_systematic_bound}
  For any systematic $(M, N, \varepsilon)$-AMD code, one has
  \[
    N_1 \ge \frac{1}{\varepsilon}
    \quad\text{and}\quad
    N_2 \ge \frac{1}{\varepsilon}.
  \]

\end{lemma}
We now recall an explicit optimal construction of systematic AMD codes.

\begin{lemma}
  \cite{cramer2013algebraic}
  \label{lem:amd_construction}
  Let $\m = (m_0, m_1, \ldots, m_{k-1}) \in \Fqk$, let $p$ be the characteristic of $\Fq$, and assume $p \nmid (k+2)$. Define
  \[
    g(x, \m) = x^{k+2} + \sum_{i=1}^{k} m_{i-1} x^i.
  \]

  Then $g$ yields a systematic $(q^k, q^{k+2}, (k+1)/q)$-AMD code.
\end{lemma}

This construction attains the lower bounds of Lemma~\ref{lem:amd_systematic_bound} with equality, and is therefore optimal among systematic AMD codes.

\section{Efficient Decoding of TGRS codes}

\label{sec:decoding_tgrs_codes}

In this section, we present a list decoding algorithm for TGRS codes, analyze its decoding radius, list size and complexity, then derive a fast unique decoding algorithm based on it.

Observe that TGRS codes are polynomial evaluation codes. Consequently, any TGRS code can be regarded as a subcode of a suitably chosen GRS code. Let $\mC_{\TGRS}(\bGa, \bv, k, \mT)$ be a TGRS code with twists $\mT = \{(t_i, h_i, \eta_i) : 1 \leq i \leq \ell\}$. We define

\[
  k' = k + \max_{1 \le i \le \ell} t_i
\]

and refer to $k'$ as the \concept{pseudo-dimension} of $\mC_{\TGRS}(\bGa, \bv, k, \mT)$. It is immediate that

\[\mP_{\TGRS}(k, \mT) \subseteq \Fq[x]_{<k'}.\]

Therefore, $\mC_{\TGRS}(\bGa, \bv, k, \mT)$ is a subcode of $\mC_{\GRS}(\bGa, \bv, k')$. This observation is formalized in the following lemma.

\begin{lemma}
  \label{lem:tgrs_subcode_grs}
  Let $\mC_{\TGRS}(\bGa, \bv, k, \mT)$ be a TGRS code with pseudo-dimension $k'$. Then
  $\mP_{\TGRS}(k, \mT) \subseteq \Fq[x]_{<k'}$, and $\mC_{\TGRS}(\bGa, \bv, k, \mT)$ is a subcode of the GRS code $\mC_{\GRS}(\bGa, \bv, k')$.
\end{lemma}

\subsection{Algorithm}

In this subsection, we describe the procedure for list decoding TGRS codes using the Guruswami--Sudan algorithm. Given a received word, we first run the Guruswami--Sudan algorithm for the code $\mC_{\GRS}(\bGa, \bv, k')$ and obtain a list of candidate polynomials. We then filter this list by retaining only those polynomials that lie in $\mP_{\TGRS}(k, \mT)$, yielding a list $\mL$. We will show that, when the Guruswami--Sudan algorithm is invoked with decoding radius $\tau$, the list $\mL$ consists exactly of all polynomials whose corresponding codewords are within distance $\tau$ of the received word. The procedure is given in Algorithm~\ref{alg:list_dec_tgrs}.

\begin{alg}
  \label{alg:list_dec_tgrs}

  List decoding algorithm for the TGRS code $\mC_{\TGRS}(\bGa, \bv, k, \mT)$.

  Input: A received word $\br \in \Fq^n$ and a decoding radius $\tau$.

  Output: A list $\mL$ of polynomials $f \in \mP_{\TGRS}(k, \mT)$ or \fail{}.

  \begin{enumerate}
    \item Compute the pseudo-dimension $k'$ of $\mC_{\TGRS}(\bGa, \bv, k, \mT)$.
    
    \item Run the Guruswami--Sudan algorithm on the GRS code $\mC_{\GRS}(\bGa, \bv, k')$ with received word $\br$ and decoding radius $\tau$. If the algorithm outputs \fail{}, output \fail{}. Otherwise, let $\mL_1$ be the list of candidate polynomials in $\Fqx_{<k'}$ returned by the algorithm.
    \label{step:list_dec_tgrs__run_gs}

    \item Initialize $\mL$ as the empty list. For each polynomial
    \[
      f = f_0 + \cdots + f_{k-1} x^{k-1} + f_k x^k + \cdots + f_{k'-1} x^{k'-1} \in \mL_1,
    \]
    compute
    \[
      f' = f_0 + \cdots + f_{k-1} x^{k-1}
           + \sum_{j=1}^{\ell} \eta_j f_{h_j} x^{k-1+t_j}.
    \]

    If $f = f'$, append $f$ to $\mL$.
    \label{step:list_dec_tgrs__check_tgrs}

    \item Output the list $\mL$.
  \end{enumerate}
\end{alg}

\subsection{Decoding Radius and List Size}

\begin{theorem}
  \label{thm:list_dec_tgrs__correctness}
  Let $\mC_{\TGRS}(\bGa, \bv, k, \mT)$ be a TGRS code with pseudo-dimension $k'$.
  Given a received word $\br$ and a decoding radius $\tau$ satisfying $\tau < n - \sqrt{n (k'-1)}$,
  Algorithm~\ref{alg:list_dec_tgrs} outputs a list $\mL$ such that
  \[
    \mL
    = \bigl\{ f \in \mP_{\TGRS}(k, \mT) : \rd\bigl(\br, \evav(f)\bigr) \le \tau \bigr\}.
  \]

  In particular, for any fixed multiplicity $s$, Algorithm~\ref{alg:list_dec_tgrs} succeeds whenever
  \[
    \tau < n - \lfloor \sqrt{n(k'-1)(1+1/s)}\rfloor.
  \]
\end{theorem}

\begin{proof}
  Let $\mL_1$ denote the list returned by the Guruswami--Sudan algorithm in
  Step~\ref{step:list_dec_tgrs__run_gs} of Algorithm~\ref{alg:list_dec_tgrs}.
  By Theorem~\ref{thm:guruswami_sudan},
  \[
    \mL_1
    = \bigl\{ f \in \Fqx_{<k'} : \rd\bigl(\evav(f), \br\bigr) \le \tau \bigr\}.
  \]

  In Step~\ref{step:list_dec_tgrs__check_tgrs}, for any polynomial $f = f_0 + \cdots + f_{k'-1} x^{k'-1} \in \mL_1$, the algorithm inserts $f$ into $\mL$ if and only if $f$ coincides with
  \[
    f'
    = f_0 + \cdots + f_{k-1} x^{k-1}
    + \sum_{j=1}^{\ell} \eta_j f_{h_j} x^{k-1+t_j}.
  \]

  By construction, $f' \in \mP_{\TGRS}(k,\mT)$. Hence, if $f$ is added to $\mL$,
  then necessarily $f=f' \in \mP_{\TGRS}(k,\mT)$.

  Conversely, assume $f \in \mP_{\TGRS}(k,\mT)$. By Definition~\ref{def:tgrs},
  the twist terms contribute only monomials of degree at least $k$. Therefore,
  every polynomial in $\mP_{\TGRS}(k,\mT)$ is uniquely determined by its first $k$
  coefficients $f_0,\ldots,f_{k-1}$ and must have the form
  \[
    f
    = f_0 + \cdots + f_{k-1} x^{k-1}
    + \sum_{j=1}^{\ell} \eta_j f_{h_j} x^{k-1+t_j}.
  \]

  In particular, $f=f'$, and thus $f$ is added to $\mL$ in
  Step~\ref{step:list_dec_tgrs__check_tgrs}. Consequently,
  \[
    f \text{ is added to } \mL
    \quad\Longleftrightarrow\quad
    f \in \mP_{\TGRS}(k,\mT).
  \]

  It follows that
  \[
    \begin{split}
      \mL ={} & \bigl\{ f \in \mL_1 : f \in \mP_{\TGRS}(k,\mT) \bigr\} \\
      ={} & \bigl\{ f \in \Fqx_{<k'} : \rd\bigl(\evav(f), \br\bigr) \le \tau \bigr\} 
         \cap \mP_{\TGRS}(k,\mT).
    \end{split}
  \]

  By Lemma~\ref{lem:tgrs_subcode_grs}, we have $\mP_{\TGRS}(k,\mT) \subseteq \Fqx_{<k'}$,
  and hence
  \[
    \mL
    = \bigl\{ f \in \mP_{\TGRS}(k,\mT) : \rd\bigl(\evav(f), \br\bigr) \le \tau \bigr\}.
  \]

  Finally, for fixed $s$, the admissible choice
  $\tau < n - \lfloor \sqrt{n(k'-1)(1+1/s)}\rfloor$ follows directly from Theorem~\ref{thm:guruswami_sudan}.
\end{proof}

Since Algorithm~\ref{alg:list_dec_tgrs} applies the Guruswami--Sudan algorithm then filters the resulting candidates, its output list is necessarily a sublist of the Guruswami--Sudan output. Consequently, the list size of our list decoding algorithm follows immediately from Lemma~\ref{lem:bound_on_list_size}.

\begin{proposition}
  \label{prop:list_size_tgrs}
  Let $\mL$ denote the list returned by Algorithm~\ref{alg:list_dec_tgrs} when decoding the TGRS code $\mC_{\TGRS}(\bGa,\bv,k,\mT)$ with length $n$ and pseudo-dimension $k'$. Then
  \[
    |\mL| \le \sqrt{ns(s+1)/(k'-1)},
  \]
  where $s$ is the multiplicity in Guruswami--Sudan algorithm. In particular, for TGRS codes of fixed rate $R = k/n$ and fixed constant multiplicity $s$, Algorithm~\ref{alg:list_dec_tgrs} output a list of constant size.
\end{proposition}

\subsection{Complexity Analysis}

We now analyze the time complexity of Algorithm~\ref{alg:list_dec_tgrs}. We first bound the cost of Step~\ref{step:list_dec_tgrs__check_tgrs}, where we test whether a candidate polynomial returned by the Guruswami--Sudan algorithm satisfies the twist constraints, i.e., whether it is a valid polynomial in $\mP_{\TGRS}(k,\mT)$.

\begin{lemma}
  \label{lem:check_tgrs_complexity}
  Let $\mC_{\TGRS}(\bGa,\bv,k,\mT)$ be a TGRS code with pseudo-dimension $k'$ and $\ell$ twists. Then Step~\ref{step:list_dec_tgrs__check_tgrs} of Algorithm~\ref{alg:list_dec_tgrs} can be implemented in time
  \[
    O\!\left(s\sqrt{n/k'} \cdot \max(\ell,\,k'-k)\right),
  \]
  where $s$ is the multiplicity in Guruswami--Sudan algorithm.
\end{lemma}

\begin{proof}
  By Proposition~\ref{prop:list_size_tgrs}, the list size satisfies $\lvert \mL_1 \rvert \le \sqrt{ns(s+1)/(k'-1)} = O(s\sqrt{n/k'})$. For a fixed candidate $f$,
  we construct
  \[
    f' = f_0 + \cdots + f_{k-1} x^{k-1}
         + \sum_{j=1}^\ell \eta_j f_{h_j} x^{k-1+t_j},
  \]
  which requires forming $\ell$ twisted contributions. We then compare the
  coefficients of degree $\ge k$ in $f$ and $f'$, i.e., at most $k'-k$
  coefficient comparisons. Thus processing a single candidate takes $O\bigl(\max(\ell,\,k'-k)\bigr)$ time, and processing all candidates in $\mL_1$ costs
  \[
    O\!\left(s\sqrt{n/k'} \cdot \max(\ell,\,k'-k)\right).
  \]
\end{proof}

In Step~\ref{step:list_dec_tgrs__run_gs}, we invoke the Guruswami--Sudan algorithm to decode the corresponding GRS code. By Theorem~\ref{thm:guruswami_sudan}, this step can be implemented in
\[
  O\!\left(\Bigl(s\frac{n}{k'}\Bigr)^{O(1)}\bigl(n\log^{2} n \log\log n\bigr)\right)
\]
time. Combining the running times of Step~\ref{step:list_dec_tgrs__run_gs} and Step~\ref{step:list_dec_tgrs__check_tgrs}, we obtain the overall complexity of Algorithm~\ref{alg:list_dec_tgrs} as stated below.

\begin{theorem}
  \label{thm:list_dec_tgrs__complexity}
  Let $\mC_{\TGRS}(\bGa,\bv,k,\mT)$ be a TGRS code of length $n$ over $\Fq$ with pseudo-dimension $k'$. Then Algorithm~\ref{alg:list_dec_tgrs} runs in time
  \[
    O\!\left(\Bigl(s\frac{n}{k'}\Bigr)^{O(1)}\bigl(n\log^2 n\log\log n + \ell\bigr)\right),
  \]
  where $s$ is the multiplicity in Guruswami--Sudan algorithm.
  In particular, for fixed-rate TGRS codes and fixed constant multiplicity $s$, Algorithm~\ref{alg:list_dec_tgrs} runs in time

  \[
    O\bigl(n \log^2 n \log\log n + \ell\bigr).
  \]

\end{theorem}

\begin{proof}
  By Theorem~\ref{thm:guruswami_sudan}, 
  Step~\ref{step:list_dec_tgrs__run_gs} of Algorithm~\ref{alg:list_dec_tgrs}
  can be implemented in time $O\!\left((s\frac{n}{k'})^{O(1)} n\log^2 n\log\log n\right)$.

  By Lemma~\ref{lem:check_tgrs_complexity}, 
  Step~\ref{step:list_dec_tgrs__check_tgrs} costs
    $O\!\left(s\sqrt{n/k'} \cdot \max(\ell,\,k'-k)\right)$.
  Moreover, $s\sqrt{n/k'}\,(k'-k) \le s\sqrt{n/k'}\cdot k' \le sn$,
  and hence this term is dominated by the interpolation and root-finding cost in
  Step~\ref{step:list_dec_tgrs__run_gs} and can be absorbed into
  $O\!\left((s\frac{n}{k'})^{O(1)}\,n\log^2 n\log\log n\right)$.
  
  Therefore, the overall running time is
  \[
    \begin{split}
      & O\!\left(\Bigl(s\frac{n}{k'}\Bigr)^{O(1)} n\log^2 n\log\log n\right) + O\!\left(s\sqrt{n/k'}\,\ell\right) \\
      ={} & O\!\left(\Bigl(s\frac{n}{k'}\Bigr)^{O(1)}\bigl(n\log^2 n\log\log n + \ell\bigr)\right).
    \end{split}
  \]

  For fixed-rate TGRS codes, we have $n/k' = O(1)$, and thus the factor
  $\bigl(s\frac{n}{k'}\bigr)^{O(1)}$ is a constant when $s=O(1)$; it can be absorbed into the
  big-$O$ notation. This yields the simplified bound
  \[
    O\bigl(n\log^2 n\log\log n + \ell\bigr).
  \]
\end{proof}

\subsection{Unique Decoding of MDS TGRS Codes}

Building on the proposed list-decoding algorithm, we further derive a unique-decoding procedure. Compared with existing methods, the resulting decoder achieves an improved running time of $O(n\log^2 n\log\log n)$ for fixed-rate MDS TGRS codes and fixed constant multiplicity $s$. Within this complexity budget, the algorithm can accommodate up to $O(n\log^2 n\log\log n)$ twists, which is asymptotically larger than what is supported by previous approaches.

To achieve the optimal unique-decoding radius, however, the pseudo-dimension must be suitably constrained. We therefore characterize the parameter regime under which optimal unique decoding is guaranteed. Our analysis shows that, compared with prior work, the proposed algorithm achieves the optimal unique-decoding radius for a broader range of pseudo-dimensions, and it also supports a larger number of twists, thereby providing more flexible parameter choices.

\begin{alg}
  \label{alg:unique_dec_tgrs}

  Unique decoding algorithm for the MDS TGRS code $\mC_{\TGRS}(\bGa, \bv, k, \mT)$.

  Input: A received word $\br \in \Fq^n$.

  Output: A polynomial $f$ or \fail{}.

  \begin{enumerate}
    \item Run Algorithm~\ref{alg:list_dec_tgrs} for the TGRS code $\mC_{\TGRS}(\bGa, \bv, k, \mT)$ with received word $\br$ and decoding radius
      \[
        \tau = \left\lfloor\frac{n - k}{2}\right\rfloor.
      \]

      If Algorithm~\ref{alg:list_dec_tgrs} outputs \fail{}, then output \fail{}. Otherwise, let $\mL$ be the list of polynomials returned by the algorithm.
      \label{step:unique_dec_tgrs__run_list_dec}

    \item If $\mL$ contains exactly one polynomial $f$, output $f$; otherwise, output \fail{}.
      \label{step:unique_dec_tgrs__output}
  \end{enumerate}
\end{alg}

The following theorem establishes the correctness of Algorithm~\ref{alg:unique_dec_tgrs}.

\begin{theorem}
  \label{thm:unique_dec_tgrs__correctness}
  Let $\mC_{\TGRS}(\bGa, \bv, k, \mT)$ be an MDS TGRS code with pseudo-dimension $k'$, and let $\br \in \Fq^n$ be a received word. 
  Assume that
  \[
    k' < \frac{(n + k)^2}{4n} + 1.
  \]

  Then Algorithm~\ref{alg:unique_dec_tgrs} outputs
  \begin{enumerate}
    \item the unique polynomial $f \in \mP_{\TGRS}(k,\mT)$ such that
      \[
        \rd\bigl(\evav(f), \br\bigr) \le \left\lfloor\frac{n-k}{2}\right\rfloor
      \]
      if such a polynomial exists, or
    \item \fail{} otherwise.
  \end{enumerate}

  In particular, for any fixed multiplicity $s$, the algorithm achieves the above optimal unique-decoding radius whenever
  \[
    k' < \frac{(n+k)^2}{4n(1+1/s)} + 1.
  \]
\end{theorem}

\begin{proof}
  We first verify that Step~\ref{step:unique_dec_tgrs__run_list_dec} never outputs \fail{}.
  By Theorem~\ref{thm:list_dec_tgrs__correctness}, Algorithm~\ref{alg:list_dec_tgrs} succeeds provided that
  \[
    \tau = \left\lfloor\frac{n-k}{2}\right\rfloor < n - \sqrt{n (k'-1)},
  \]
  which is implied by the assumption $k' < \frac{(n+k)^2}{4n} + 1$, and therefore
  Step~\ref{step:unique_dec_tgrs__run_list_dec} does not return \fail{}.

  By Theorem~\ref{thm:list_dec_tgrs__correctness}, the list $\mL$ produced in
  Step~\ref{step:unique_dec_tgrs__run_list_dec} consists precisely of all twist polynomials
  $f \in \mP_{\TGRS}(k,\mT)$ whose evaluations lie within Hamming distance
  $\lfloor\frac{n-k}{2}\rfloor$ of $\br$. Since $\mC_{\TGRS}(\bGa,\bv,k,\mT)$ is an $[n,k]$ MDS code,
  its minimum distance equals $d = n-k+1$, and thus the unique-decoding radius is
  $\left\lfloor\frac{d-1}{2}\right\rfloor = \left\lfloor\frac{n-k}{2}\right\rfloor$.
  Consequently, any Hamming ball of radius $\left\lfloor\frac{n-k}{2}\right\rfloor$ contains at most one codeword of $\mC_{\TGRS}$, and hence $\lvert \mL \rvert \le 1$.

  Finally, Algorithm~\ref{alg:unique_dec_tgrs} outputs the unique polynomial $f$ when $\lvert \mL \rvert = 1$,
  and outputs \fail{} when $\lvert \mL \rvert = 0$, which establishes the claim.

  For fixed $s$, Theorem~\ref{thm:list_dec_tgrs__correctness} guarantees successful list decoding whenever
  \[
    \tau = \left\lfloor\frac{n-k}{2}\right\rfloor < n - \lfloor \sqrt{n(k'-1)(1+1/s)}\rfloor.
  \]

  This inequality holds under the condition $k' < \frac{(n+k)^2}{4n(1+1/s)} + 1$, completing the proof.

\end{proof}

The running time of Algorithm~\ref{alg:unique_dec_tgrs} is the same as that of Algorithm~\ref{alg:list_dec_tgrs}.

\begin{remark}
  We estimate the maximal number of twists that can be supported while still
  satisfying the applicability condition of
  Algorithm~\ref{alg:unique_dec_tgrs}. Throughout, we consider a family of TGRS
  codes of fixed rate $R = k/n$.

  By Theorem~\ref{thm:unique_dec_tgrs__correctness}, 
  Algorithm~\ref{alg:unique_dec_tgrs} applies whenever
  \[
    k' < \frac{(n+k)^2}{4n} + 1
    \qquad\Longleftrightarrow\qquad
    k' - k < \frac{(n-k)^2}{4n} + 1.
  \]
  To maximize the admissible number of twists, we may therefore choose
  \[
    k' - k = \left\lceil \frac{(n-k)^2}{4n} \right\rceil.
  \]
  For fixed rate $R = k/n$, substituting $k = Rn$ yields
  \[
    k' - k
      = \left\lceil \frac{(1-R)^2}{4}\,n \right\rceil.
  \]

  Since the pairs $(t_i,h_i)$ are required to be pairwise distinct, with
  $t_i$ taking at most $k'-k$ possible values and $h_i$ taking $k$ possible
  values, we obtain the combinatorial bound
  \[
    \ell \le k\,(k'-k)
      = Rn\left\lceil \frac{(1-R)^2}{4}\,n \right\rceil
      = O(n^2).
  \]
  Thus, for fixed-rate TGRS codes, our decoder can accommodate a number of
  twists growing quadratically in~$n$, which is asymptotically optimal, in
  contrast to existing decoding algorithms that are restricted to the regime
  $\ell = 1$.
\end{remark}

\begin{remark}
  By Theorem~\ref{thm:list_dec_tgrs__complexity}, for fixed-rate MDS
  TGRS codes and fixed constant multiplicity $s$, our unique decoding algorithm for TGRS codes runs in time
  \[
    O\bigl(n \log^2 n \log\log n + \ell\bigr).
  \]

  In particular, when the number of twists satisfies
  $\ell = O\!\bigl(n \log^2 n \log\log n\bigr)$, the running time becomes
  \[
    O\bigl(n \log^2 n \log\log n\bigr).
  \]

  This improves asymptotically over the previously best-known $O(n^2)$ decoder
  of~\cite{wang2025improved}, while allowing the number of twists $\ell$ to grow
  beyond the constant-twist regime.
\end{remark}

\section{Efficient Decoding of TGRS codes with AMD codes}

\label{sec:decoding_tgrs_amd}

In the previous section, we proposed a list-decoding algorithm for TGRS codes. In many applications, however, one ultimately needs a single decoded message rather than a list of candidates. Guruswami and Smith~\cite{guruswami2010codes} introduced an AMD-code-based method to reliably select the intended polynomial from the list returned by a list decoder for GRS codes. Their original treatment focuses on the binary case; here we extend the approach to the $q$-ary setting. In addition, we provide an explicit complexity analysis for the resulting AMD-assisted encoding and decoding procedures for TGRS codes.

\subsection{Algorithm}

The core idea is as follows. In the encoding phase, we first map the message to an AMD-augmented information vector and then encode this augmented vector using a TGRS code. In the decoding phase, we apply the TGRS list decoding algorithm to the received word to obtain a list of candidate AMD codewords, and then invoke the AMD verification property to identify the unique valid AMD-augmented information vector. Finally, we recover the original message from this verified vector.

We employ the systematic AMD code from Lemma~\ref{lem:amd_construction}. Two issues remain.

The first issue is that this AMD code has a fixed error parameter $\varepsilon$ once the message length $k$ and field size $q$ are fixed. If the decoding list is large, the probability that none of the invalid AMD codewords passes verification may become unacceptably small. To mitigate this, we amplify the AMD security by working over an extension field and aggregating message symbols into larger blocks.

In the encoding phase, let $\m$ be a message of length $k$ over $\Fq$. Partition $\m$ into $\lceil k/b \rceil$ blocks of length $b$. If the last block has fewer than $b$ symbols, pad it with zeros. Each block in $\Fq^b$ is then identified with a single symbol of the extension field $\Fqb$, yielding a compressed message $\mpr \in \Fqb^{\lceil k/b \rceil}$.

We apply the AMD construction of Lemma~\ref{lem:amd_construction} over $\Fqb$ to obtain a random seed $x' \in \Fqb$ and a tag $t' = g(x', \mpr) \in \Fqb$.

Next, we expand $x'$ and $t'$ back to vectors $x,t \in \Fqbv$. Finally, we form the AMD-augmented information vector
\[
  \bg = (\m \,\|\, x \,\|\, t) \in \Fq^{k + 2b}
\]
by concatenating the original message $\m$ with $x$ and $t$, and then encode $\bg$ using the TGRS code.

In the decoding phase, the AMD verification is performed by first extracting $(\hat{\m}, \hat{x}, \hat{t})$ from each candidate polynomial, recompressing $\hat{\m}$ and $\hat{x}$ into $\hat{\mpr} \in \Fqb^{\lceil k/b \rceil}$ and $\hat{x}' \in \Fqb$, recomputing the tag $\hat{t}'_{0} = g(\hat{x}', \hat{\mpr})$, and finally checking whether $\hat{t}'_0$ coincides with the recompressed version of $\hat{t}$. Only candidates that pass this check are retained. This modification preserves the AMD security guarantees while effectively reducing the error probability per candidate.

The second issue is that Lemma~\ref{lem:amd_construction} requires that the characteristic $p$ of $\Fq$ does not divide $k+2$. This can be enforced by a minor change in the message length: whenever $p \mid k+2$, we append one dummy zero symbol to the message, thereby increasing $k$ by $1$ and ensuring $p \nmid (k+1)+2$. This does not affect the information content of the message but restores the applicability of the AMD construction.

We fix an $\F_q$-basis of $\F_{q^b}$ once and for all; throughout, \emph{pack/unpack} refers to identification between $\F_q^b$ and $\F_{q^b}$ under this basis. We summarize the computation of the AMD tag in the following algorithm.

\begin{alg}
  \label{alg:amd_tag}

  Computation of the AMD tag using the systematic AMD code of Lemma~\ref{lem:amd_construction}.

  Input: A message $\m = (m_0, m_1, \ldots, m_{k-1}) \in \Fqk$, a block size $b$, and a seed $x' \in \Fqb$.

  Output: An AMD tag $t' \in \Fqb$.

  \begin{enumerate}
    \item Partition $\m$ into $r_0=\lceil k/b\rceil$ consecutive blocks of length $b$ (padding the last block with zeros if needed), and pack them into
    \[
      \mpr=(m'_0,\ldots,m'_{r_0-1})\in\F_{q^b}^{r_0}.
    \]

    \item Let $p=\mathrm{char}(\F_q)$. If $p\mid(r_0+2)$, append one zero symbol to $\mpr$ and set $r=r_0+1$; otherwise set $r=r_0$.
    \label{step:amd_tag__append_zero_to_m}
    
    \item Compute and output the AMD tag
      \[
        t' = g(x', \mpr) = (x')^{r+2} + \sum_{i=1}^{r} m'_{i-1} (x')^i \in \Fqb.
      \]
    \label{step:amd_tag__tag_computation}
  \end{enumerate}
\end{alg}

We then present the encoding and decoding algorithms below.

\begin{alg}
  \label{alg:encoding_tgrs_amd}

  Encoding algorithm for the TGRS code $\mC_{\TGRS}(\bGa, \bv, k+2b, \mT)$ using AMD preencoding with block size $b$.

  Input: A message $\m = (m_0, m_1, \ldots, m_{k-1}) \in \Fqk$.

  Output: A codeword $\bc \in \mC_{\TGRS}(\bGa, \bv, k+2b, \mT)$.

  \begin{enumerate}
    \item Sample a seed $x' \in \Fqb$ uniformly at random.
    \label{step:encoding_tgrs_amd__sample_seed}
    \item Run Algorithm~\ref{alg:amd_tag} on input $(\m, b, x')$ to obtain the AMD tag $t' \in \Fqb$.
    \item Unpack $x'$ and $t'$ into $x, t \in \Fqbv$, then form the AMD-augmented information vector
      \[
        \bg = (\m \,\|\, x \,\|\, t) \in \Fq^{k+2b}.
      \]
    \label{step:encoding_tgrs_amd__generate_g}
    \item Encode $\bg$ with the TGRS encoder of $\mC_{\TGRS}(\bGa, \bv, k+2b, \mT)$ and output the resulting codeword $\bc$.
  \end{enumerate}
\end{alg}

\begin{alg}
  \label{alg:decoding_tgrs_amd}

  Decoding algorithm for the TGRS code $\mC_{\TGRS}(\bGa, \bv, k+2b, \mT)$ using AMD preencoding with block size~$b$.

  Input: A received word $\br \in \Fq^n$ and a decoding radius $\tau$.

  Output: A message $\m \in \Fq^k$ or \fail{}.

  \begin{enumerate}
    \item Run Algorithm~\ref{alg:list_dec_tgrs} for the TGRS code $\mC_{\TGRS}(\bGa, \bv, k+2b, \mT)$ on input $(\br, \tau)$. If it outputs \fail{}, output \fail{}. Otherwise, let $\mL_1$ be the returned list.
    \label{step:decoding_tgrs_amd__run_list_dec}
    \item Initialize $\mL$ as an empty list.
    \item For each candidate polynomial in $\mL_1$, do:
    \label{step:decoding_tgrs_amd__amd_check}
      \begin{enumerate}
        \item Extract its first $k+2b$ coefficients, which correspond to the AMD-preencoded message, and parse them as
          \[
            (\hat{\m}, \hat{x}, \hat{t}) \in \Fq^k \times \Fq^b \times \Fq^b.
          \]
        \item Pack $\hat{x}$ and $\hat{t}$ into elements $\hat{x}'$ and $\hat{t}'$ in $\F_{q^b}$. Run Algorithm~\ref{alg:amd_tag} on input $(\hat{\m}, b, \hat{x}')$ to obtain the AMD tag $\hat{t}'_0 \in \Fqb$. If $\hat{t}'_0 = \hat{t}'$, append $\hat{\m}$ to $\mL$.
      \end{enumerate}
    \item If $\mL$ contains exactly one message vector, output this vector as $\m$. Otherwise, output \fail{}.
  \end{enumerate}
\end{alg}

\subsection{Error probability}

Since the underlying AMD code used in Algorithm~\ref{alg:decoding_tgrs_amd} has a non-zero error probability, the overall decoding procedure may also fail. However, by choosing the block length~$b$ sufficiently large, this error probability can be made negligibly small.

A subtlety is that the AMD soundness bound in Definition~\ref{def:amd} is stated for a fixed manipulation $\delta$. In our setting, the manipulation arises from the channel error affecting a TGRS codeword. We therefore need to argue that, once the error vector on the TGRS codeword is fixed, the induced manipulation on the AMD-augmented information vector is also fixed, so that the AMD bound applies.

The following theorem bounds the decoding radius and error probability of Algorithm~\ref{alg:decoding_tgrs_amd}.

\begin{theorem}
  \label{thm:decoding_tgrs_amd__correctness}
  Let $\mC_{\TGRS}(\bGa, \bv, k+2b, \mT)$ be a TGRS code with pseudo-dimension $k'$. Suppose a message $\m$ is encoded by Algorithm~\ref{alg:encoding_tgrs_amd} into a codeword $\bc \in \mC_{\TGRS}(\bGa, \bv, k+2b, \mT)$. Let $\br$ be a received word and let $\tau$ be a decoding radius such that
  \[
    \tau < n - \sqrt{n (k'-1)}.
  \]

  If $\rd(\br,\bc) \le \tau$, then Algorithm~\ref{alg:decoding_tgrs_amd} outputs $\m$ with probability at least
  \[
    1 - \frac{\lceil k/b\rceil + 2}{q^b}\left(\sqrt{\frac{ns(s+1)}{(k'-1)}} - 1\right),
  \]
  where $s$ is the multiplicity in Guruswami--Sudan algorithm.

  In particular, for fixed constant multiplicity $s$, the algorithm succeeds whenever
  \[
    \tau < n - \lfloor\sqrt{n(k'-1)(1+1/s)}\rfloor.
  \]
\end{theorem}

\begin{proof}
  First observe that $\m$ is always contained in the final list $\mL$. Let $\bg$ be the AMD-augmented information vector produced in Step~\ref{step:encoding_tgrs_amd__generate_g} of Algorithm~\ref{alg:encoding_tgrs_amd}. By Theorem~\ref{thm:list_dec_tgrs__correctness}, since $\rd(\br,\bc)\le\tau$ and $\tau < n - \sqrt{n(k'-1)}$ (or $\tau < n - \lfloor\sqrt{n(k'-1)(1+1/s)}\rfloor$ for fixed $s$), the polynomial corresponding to $\bg$ must appear in $\mL_1$. Moreover, $\bg$ is a valid AMD-augmented information vector by construction, so it necessarily passes the AMD verification in Algorithm~\ref{alg:decoding_tgrs_amd}. Hence $\m$ is added to~$\mL$.

  It remains to bound the probability that any other element of $\mL_1$ passes the AMD check. Write $\bg_1 = \bg,\;\bg_2,\ldots,\bg_L$ for the AMD-augmented information vectors corresponding to the polynomials in $\mL_1$, and let $\bc_1 = \bc,\;\bc_2,\ldots,\bc_L$ be the associated TGRS codewords. Let $\be = \br - \bc$ be the channel error vector. By linearity of $\mC_{\TGRS}$, the set of codewords within distance~$\tau$ of~$\be$ is exactly
  \[
    \bc'_1 = \bc_1 - \bc = 0,\;\bc'_2 = \bc_2 - \bc,\ldots,\bc'_L = \bc_L - \bc.
  \]

  The corresponding information vectors are
  \[
    \delta_1 = \bg_1 - \bg = 0,\;\delta_2 = \bg_2 - \bg,\ldots,\delta_L = \bg_L - \bg.
  \]

  Each $\delta_j$ is completely determined by $\be$ and the code and is, in particular, independent of the randomness used in the AMD encoding (the choice of $x'$). Thus, from the AMD-code viewpoint, each $\delta_j$ is a fixed manipulation applied to~$\bg$.

  By Lemma~\ref{lem:amd_construction} (applied over~$\Fqb$), for any non-zero manipulation $\delta_j$ the probability that the manipulated AMD-augmented information vector passes verification is at most $(r+1)/q^b$. By Algorithm~\ref{alg:amd_tag}, we have $r \le \left\lceil k/b \right\rceil + 1$, and therefore
  \[
  \begin{split}
    & \Pr[\text{a fixed manipulated vector passes the AMD check}] \\
    \le{} & \frac{\lceil k/b\rceil + 2}{q^b}.
  \end{split}
  \]

  There are $L-1$ manipulated candidates $\bg_2,\ldots,\bg_L$. By a union bound, the probability that at least one of them passes the AMD check is at most $\frac{\lceil k/b\rceil + 2}{q^b}\,(L-1)$.
  By Proposition~\ref{prop:list_size_tgrs}, $L \le \sqrt{ns(s+1)/(k'-1)}$, and hence, with probability at least
  \[
    1 - \frac{\lceil k/b\rceil + 2}{q^b}\left(\sqrt{\frac{ns(s+1)}{(k'-1)}} - 1\right),
  \]
  the only element in $\mL$ is~$\m$, so Algorithm~\ref{alg:decoding_tgrs_amd} outputs~$\m$.
\end{proof}

\begin{remark}
  We now give a convenient upper bound on the error term in
  Theorem~\ref{thm:decoding_tgrs_amd__correctness}. We have
  \[
  \begin{split}
    & \frac{\lceil k/b\rceil + 2}{q^b}\Bigl(\sqrt{\frac{ns(s+1)}{(k'-1)}} - 1\Bigr) \\
    \le{} & \frac{k/b + 3}{q^b}\Bigl((s+1)\sqrt{\frac{n}{(k-1)}} - 1\Bigr) \\
    \le{} & \frac{k/b + 3}{q^b}\,(s+1)\sqrt{\frac{n}{(k-1)}}.
  \end{split}
  \]

  Using $k \le n \le q$ gives
  \[
  \begin{split}
    & \frac{\lceil k/b\rceil + 2}{q^b}\Bigl(\sqrt{\frac{ns(s+1)}{(k'-1)}} - 1\Bigr) \\
    \le{} & \frac{(s+1) \sqrt{kn(1+\frac{1}{k-1})}/b}{q^{b}} + \frac{3 (s+1) \sqrt{\frac{n}{k-1}}}{q^{b}} \\
    \le{} & \frac{\sqrt{2}(s+1)}{b\,q^{b-1}} + \frac{3(s+1)}{q^{b-1}} \\
    \le{} & \frac{5(s+1)}{q^{b-1}},
  \end{split}
  \]
  since $b \ge 1$, $k \ge 2$ and $1/(k-1) \le 1$. Thus, the error probability decreases in exponent of $1/q$ when $b$ increases. In particular, we can achieve error probability $5(s+1)/q$ (which is linear in $1/q$ for fixed constant multiplicity $s$) by taking $b = 2$.
\end{remark}

\subsection{Complexity Analysis}

We now analyze the time complexity of the algorithms. We begin with the cost of computing the AMD tag.

\begin{lemma}
  \label{lem:amd_complexity}
  Let $k$ be the message length and $b$ the block size. Counting operations over $\Fq$ as basic operations, Algorithm~\ref{alg:amd_tag} runs in time
  \[
    O\bigl(k\log b\bigr).
  \]
\end{lemma}

\begin{proof}
  The dominant cost arises from evaluating the polynomial
  \[
    g(x', \m') = (x')^{r+2} + \sum_{i=1}^{r} m'_{i-1}(x')^i \in \Fqb,
  \]
  where $r = O(k/b)$. This requires $O(r)$ additions and $O(r)$ multiplications in $\Fqb$. Using a fast multiplication algorithm over $\Fqb$, one multiplication can be performed in $O(b\log b)$ operations over $\Fq$. Hence the time complexity is
  \[
    O\bigl(r\cdot b\log b\bigr)
      = O\!\left(\frac{k}{b}\cdot b\log b\right)
      = O(k\log b).
  \]
\end{proof}

A standard encoder for TGRS codes based on fast multipoint polynomial evaluation runs in time $O\bigl(n\log^2 n\bigr)$\cite{borodin1974fast}. Since the AMD preencoding only introduces an additional preprocessing step whose cost is asymptotically dominated by the evaluation phase, it does not affect the overall complexity. This yields the following proposition.

\begin{proposition}
  Let $n$ be the code length. Counting operations over $\Fq$ as basic operations, Algorithm~\ref{alg:encoding_tgrs_amd} runs in time
  \[
    O\bigl(n\log^2 n\bigr).
  \]
\end{proposition}

We finally analyze the complexity of our decoding algorithm.

\begin{theorem}
  Let $\mC_{\TGRS}(\bGa,\bv,k,\mT)$ be a TGRS code of length $n$ over $\Fq$ with pseudo-dimension $k'$. Counting operations over $\Fq$ as basic operations, Algorithm~\ref{alg:decoding_tgrs_amd} runs in time
  \[
    O\!\left(\Bigl(s\frac{n}{k'}\Bigr)^{O(1)}\bigl(n\log^2 n\log\log n + \ell\bigr)\right),
  \]
  where $s$ is the multiplicity in Guruswami--Sudan algorithm.
\end{theorem}

\begin{proof}
  The running time has two contributions: the underlying TGRS list decoding and the AMD verification.

  By Theorem~\ref{thm:list_dec_tgrs__complexity}, the list decoding part costs
  \[
    O\!\left(\Bigl(s\frac{n}{k'}\Bigr)^{O(1)}\bigl(n\log^2 n\log\log n + \ell\bigr)\right).
  \]

  For the AMD verification, each candidate in $\mL_1$ is checked in
  $O(k\log b)$ time by Lemma~\ref{lem:amd_complexity}. Since
  $\lvert \mL_1 \rvert \le \sqrt{ns(s+1)/(k'-1)} = O(s\sqrt{n/k'})$ by Proposition~\ref{prop:list_size_tgrs}, the total AMD cost is $O\!\left(s\sqrt{n/k'}\,k\log b\right)$. Using $k \le n$ and $b \le k' \le n$, this term is absorbed into $O\!\left((s\frac{n}{k'})^{O(1)}n\log^2 n\log\log n\right)$, and hence does not affect the overall asymptotic complexity.

  Therefore, Algorithm~\ref{alg:decoding_tgrs_amd} runs in time
  \[
    O\!\left(\Bigl(s\frac{n}{k'}\Bigr)^{O(1)}\bigl(n\log^2 n\log\log n + \ell\bigr)\right),
  \]
  as claimed.
\end{proof}

To better understand the procedure of the AMD-assisted decoding algorithm, we present a numerical example below that further demonstrates its correctness. As observed in~\cite{mceliece2003guruswami}, the output list is typically of size one in most instances. Here we instead construct a received word whose Hamming ball of the prescribed radius contains two codewords, and show how the AMD mechanism is used to disambiguate the candidates and recover the intended message.

\begin{example}
  \label{exm:tgrs_with_amd}
  Consider the TGRS code $\mC_{\TGRS}(\{0,1,\ldots,22\}, \boldsymbol{1}, 5, \{(1,1,1)\})$ over $\F_{23}$. 
  Suppose the message $(6, 8, 18)$ is encoded into the codeword
  \[
  \begin{split}
    (& 6, 5, 12, 6, 10, 16, 2, 18, 19, 20, 21, \\
     &1, 4, 18, 15, 14, 6, 17, 10, 17, 18, 4, 17).
  \end{split}
  \]

  Assume the transmission is affected by the error vector of Hamming weight $\lceil n - \sqrt{n(k'-1)}\rceil - 1 = 12$:
  \[
  \begin{split}
    (&17, 8, 0, 10, 0, 8, 0, 0, 0, 5, 0, 0, \\
     &0, 0, 4, 11, 11, 0, 16, 22, 6, 0, 14).
  \end{split}
  \]

  Running the list decoding algorithm on the received word yields two candidate
  AMD-augmented messages:
  \[
    (0, 4, 20, 6, 2), \qquad (6, 8, 18, 1, 10).
  \]

  The first candidate fails the AMD verification, while the second one passes and
  corresponds to the originally transmitted message.
\end{example}

\section{Efficient Decoding of Roth--Lempel Codes}

\label{sec:decoding_rl_codes}

In this section, we present a list decoding algorithm for Roth--Lempel codes and analyze its decoding radius, list size and time complexity. We then describe a corresponding unique decoding algorithm, followed by an AMD-assisted list decoding algorithm.

Firstly, we have the following observation which is important for our decoding.

\begin{lemma}
  \label{lem:rl_punctured_grs}
  Let $\mC_{\RL}(\bGa, \bv, k, \delta)$ be a Roth--Lempel code as in Definition~\ref{def:rl}, and let $\bv'$ be the vector obtained from $\bv$ by deleting its last coordinate. Then
  \[
    \mC_{\GRS}(\bGa, \bv', k)
  \]
  is exactly the punctured code of $\mC_{\RL}(\bGa, \bv, k, \delta)$ obtained by removing the last coordinate.
\end{lemma}

\begin{proof}
  By Definition~\ref{def:rl}, a generator matrix of $\mC_{\RL}(\bGa, \bv, k, \delta)$ is
  \[
    G = \begin{bmatrix}
      1 & 1 & \cdots & 1 & 0 \\
      \Ga_1 & \Ga_2 & \cdots & \Ga_{n-1} & 0 \\
      \vdots & \vdots & \ddots & \vdots & \vdots \\
      \Ga_1^{k-2} & \Ga_2^{k-2} & \cdots & \Ga_{n-1}^{k-2} & 1 \\
      \Ga_1^{k-1} & \Ga_2^{k-1} & \cdots & \Ga_{n-1}^{k-1} & \delta
    \end{bmatrix}
    \cdot \diag(v_1,\ldots,v_n).
  \]

  Obviously, deleting the last column of $G$ yields a generator matrix of the GRS code $\mC_{\GRS}(\bGa, \bv', k)$. Thus, puncturing $\mC_{\RL}(\bGa, \bv, k, \delta)$ at the last coordinate yields $\mC_{\GRS}(\bGa, \bv', k)$.
\end{proof}

\subsection{Algorithm}
Based on the relationship between Roth--Lempel codes and GRS codes, we can employ the Guruswami--Sudan algorithm to decode Roth--Lempel codes. Given a received word $\br \in \Fq^n$, we first puncture its last coordinate to obtain
\[
  \br' \in \Fq^{n-1}.
\]

We then run the Guruswami--Sudan algorithm for the GRS code $\mC_{\GRS}(\bGa, \bv', k)$ with received word $\br'$ and decoding radius $\tau$, obtaining a list of candidate polynomials. Each candidate is re-encoded as a codeword of $\mC_{\RL}(\bGa, \bv, k, \delta)$, and we retain only those codewords whose distance to $\br$ does not exceed~$\tau$. The overall list decoding procedure is summarized in Algorithm~\ref{alg:list_dec_rl}.

\begin{alg}
  \label{alg:list_dec_rl}

  List decoding algorithm for the Roth--Lempel code $\mC_{\RL}(\bGa, \bv, k, \delta)$.

  Input: A received word $\br \in \Fq^n$ and a decoding radius $\tau$.

  Output: A list $\mL$ of polynomials $f \in \Fqxk$ or \fail{}.

  \begin{enumerate}
    \item Form $\br' \in \Fq^{n-1}$ by puncturing the last coordinate of $\br$.

    \item Run the Guruswami--Sudan algorithm on the GRS code $\mC_{\GRS}(\bGa, \bv', k)$ with received word $\br'$ and decoding radius $\tau$. If the algorithm outputs \fail{}, then output \fail{}. Otherwise, let $\mL_1$ be the resulting list of candidate polynomials in $\Fqx_{<k}$.
    \label{step:list_dec_rl__run_gs}

    \item Initialize $\mL$ as the empty list. For each polynomial $f$, encode $f$ with $\mC_{\RL}(\bGa, \bv, k, \delta)$ to obtain a codeword $\bc$. If
      \[
        \rd(\br, \bc) \le \tau,
      \]
      then append $f$ to $\mL$.
      \label{step:list_dec_rl__check_radius}

    \item Output the list $\mL$.
  \end{enumerate}

\end{alg}

\subsection{Decoding Radius and List Size}

\begin{theorem}
  \label{thm:list_dec_rl__correctness}
  Let $\mC_{\RL}(\bGa, \bv, k, \delta)$ be a Roth--Lempel code of length $n$ over $\Fq$, and let $\enc(f)$ denote the codeword of $\mC_{\RL}(\bGa, \bv, k, \delta)$ corresponding to a message polynomial $f \in \Fqxk$. 
  Given a received word $\br \in \Fq^n$ and a decoding radius $\tau$ satisfying
  \[
    \tau < (n - 1) - \sqrt{(n - 1)(k - 1)},
  \]
  Algorithm~\ref{alg:list_dec_rl} outputs a list $\mL$ such that
  \[
    \mL
      = \bigl\{ f \in \Fqxk : \rd\bigl(\enc(f), \br\bigr) \le \tau \bigr\}.
  \]

  In particular, for any fixed multiplicity $s$, the algorithm succeeds whenever
  \[
    \tau < (n-1) - \lfloor \sqrt{(n-1)(k-1)(1+1/s)} \rfloor.
  \]
\end{theorem}

\begin{proof}
  Let $\mL_1$ be the list returned by the Guruswami--Sudan algorithm in
  Step~\ref{step:list_dec_rl__run_gs} of Algorithm~\ref{alg:list_dec_rl}. 
  By Theorem~\ref{thm:guruswami_sudan} (applied to the GRS code
  $\mC_{\GRS}(\bGa,\bv',k)$ of length $n-1$), we have
  \[
    \mL_1
      = \bigl\{ f \in \Fqxk : \rd\bigl(\evavp(f), \br'\bigr) \le \tau \bigr\},
  \]
  where $\br'$ is the punctured word obtained from $\br$ by deleting its last
  coordinate.

  Fix $f \in \Fqxk$ and let $\bc = \enc(f)$ be the corresponding codeword in
  $\mC_{\RL}(\bGa, \bv, k, \delta)$. Let $\bc'$ be the punctured word obtained
  from $\bc$ by deleting its last coordinate. If $\rd(\bc, \br) \le \tau$, then clearly 
  
  \[
  \rd(\bc', \br') \le \tau.
  \]
  
  By Lemma~\ref{lem:rl_punctured_grs}, puncturing
  $\mC_{\RL}(\bGa, \bv, k, \delta)$ at the last coordinate yields the GRS code
  $\mC_{\GRS}(\bGa, \bv', k)$, and $\bc'$ is exactly the evaluation vector
  $\evavp(f)$ in this GRS code. Hence $f \in \mL_1$, and therefore
  \[
    \bigl\{ f \in \Fqxk : \rd\bigl(\enc(f), \br\bigr) \le \tau \bigr\}
      \subseteq \mL_1.
  \]

  In Step~\ref{step:list_dec_rl__check_radius} of
  Algorithm~\ref{alg:list_dec_rl}, the list $\mL$ is formed by retaining
  precisely those $f \in \mL_1$ whose corresponding codeword $\enc(f)$ lies
  within distance~$\tau$ of~$\br$, i.e.,
  \[
    \mL
      = \mL_1 \cap
        \bigl\{ f \in \Fqxk : \rd\bigl(\enc(f), \br\bigr) \le \tau \bigr\}.
  \]

  Combining this with the inclusion above yields
  \[
    \mL
      = \bigl\{ f \in \Fqxk : \rd\bigl(\enc(f), \br\bigr) \le \tau \bigr\},
  \]
  as claimed.

  Finally, for fixed $s$, the admissible choice
  $\tau < (n-1) - \lfloor \sqrt{(n-1)(k-1)(1+1/s)} \rfloor$ follows directly from Theorem~\ref{thm:guruswami_sudan}.
\end{proof}

The bound on the output list size follows directly from Lemma~\ref{lem:bound_on_list_size}.

\begin{proposition}
  \label{prop:list_size_rl}
  Let $\mL$ denote the list returned by Algorithm~\ref{alg:list_dec_rl} when decoding the Roth--Lempel code $\mC_{\RL}(\bGa, \bv, k, \delta)$ of length $n$. Then
  \[
    |\mL| \le \sqrt{(n-1)s(s+1)/(k-1)},
  \]
  where $s$ is the multiplicity in Guruswami--Sudan algorithm.

  In particular, for fixed-rate Roth--Lempel codes and fixed constant multiplicity $s$, Algorithm~\ref{alg:list_dec_rl} outputs a list of constant size.
\end{proposition}

\subsection{Complexity Analysis}

We then analyze the time complexity of Algorithm~\ref{alg:list_dec_rl}.

\begin{theorem}
  \label{thm:list_dec_rl__complexity}
  Let $\mC_{\RL}(\bGa, \bv, k, \delta)$ be a Roth--Lempel code of length $n$ over $\Fq$. Then Algorithm~\ref{alg:list_dec_rl} runs in time
  \[
    O\!\left(\Bigl(s\frac{n}{k}\Bigr)^{O(1)} n\log^2 n\log\log n\right).
  \]
\end{theorem}

\begin{proof}
  By Theorem~\ref{thm:guruswami_sudan}, applied to the GRS code
  $\mC_{\GRS}(\bGa,\bv',k)$ of length $n-1$, 
  Step~\ref{step:list_dec_rl__run_gs} runs in time
  \[
  \begin{split}
    & O\!\left(\Bigl(s\frac{n-1}{k}\Bigr)^{O(1)} (n-1)\log^2(n-1)\log\log(n-1)\right) \\
    ={} & O\!\left(\Bigl(s\frac{n}{k}\Bigr)^{O(1)} n\log^2 n\log\log n\right).
  \end{split}
  \]

  In Step~\ref{step:list_dec_rl__check_radius}, we process at most $\sqrt{(n-1)s(s+1)/(k-1)} = O(s\sqrt{n/k})$
  polynomials by Lemma~\ref{lem:bound_on_list_size}. For a fixed candidate $f$, we compute the associated
  Roth--Lempel codeword $\bc = \enc(f)$ and its Hamming distance to $\br$.

  Encoding a Roth--Lempel codeword amounts to evaluating a degree-$<k$
  polynomial at $n-1$ points (the first $n-1$ coordinates) and then computing
  the last coordinate. Using fast multipoint evaluation, this costs
  $O(n\log^2 n)$ field operations, while computing the last coordinate takes
  $O(1)$ additional operations. The Hamming distance to $\br$ can then be
  computed in $O(n)$ time. Thus each candidate is processed in
  $O(n\log^2 n)$ time, and the overall cost of
  Step~\ref{step:list_dec_rl__check_radius} is
  \[
    O\!\left(s\sqrt{n/k} \, n\log^2 n\right).
  \]

  This term is asymptotically dominated by
  $O\!\bigl((sn/k)^{O(1)} n\log^2 n\log\log n\bigr)$, so the total running time
  of Algorithm~\ref{alg:list_dec_rl} is
  \[
    O\!\left(\Bigl(s\frac{n}{k}\Bigr)^{O(1)} n\log^2 n\log\log n\right),
  \]
  as claimed.
\end{proof}

The following corollary follows directly from Theorem~\ref{thm:list_dec_rl__complexity}.

\begin{corollary}
  \label{cor:list_dec_rl__complexity_fixed_rate}
  Let $\mC_{\RL}(\bGa, \bv, k, \delta)$ be a Roth--Lempel code of fixed rate $R = k/n$ and fix the multiplicity $s$ as a constant. Then Algorithm~\ref{alg:list_dec_rl} runs in time
  \[
    O\!\left(n\log^2 n\log\log n\right).
  \]
\end{corollary}

\subsection{Unique Decoding of MDS Roth--Lempel Codes}

The list decoding algorithm for Roth--Lempel codes can also be used for unique decoding. We present the algorithm below and analyze the conditions for the algorithm to work.

\begin{alg}
  \label{alg:unique_dec_rl}

  Unique decoding algorithm for MDS Roth--Lempel code $\mC_{\RL}(\bGa, \bv, k, \delta)$.

  Input: A received word $\br \in \Fq^n$.

  Output: A polynomial $f$ or \fail{}.

  \begin{enumerate}
    \item Run Algorithm~\ref{alg:list_dec_rl} on the Roth--Lempel code $\mC_{\RL}(\bGa, \bv, k, \delta)$ with received word $\br$ and decoding radius $\tau = \lfloor\frac{n - k}{2}\rfloor$. If the algorithm outputs \fail{}, output \fail{}. Else, let $\mL$ be the list of polynomials returned by the algorithm.
    \label{step:unique_dec_rl__run_list_dec}

    \item If there is a unique polynomial $f$ in $\mL$, output $f$. Else, output \fail{}.
    
    \label{step:unique_dec_rl__output}
  \end{enumerate}
\end{alg}
\begin{theorem}
  \label{thm:unique_dec_rl__correctness}
  Let $\mC_{\RL}(\bGa, \bv, k, \delta)$ be an MDS Roth--Lempel code of length $n$ over $\Fq$, and let $\enc(f)$ denote the codeword of $\mC_{\RL}(\bGa, \bv, k, \delta)$ corresponding to a message polynomial $f \in \Fqxk$. 

  Assume that $k < n$. Then Algorithm~\ref{alg:unique_dec_rl} outputs
  \begin{enumerate}
    \item the unique polynomial $f \in \Fqxk$ such that
      \[
        \rd\bigl(\enc(f), \br\bigr) \le \left\lfloor\frac{n-k}{2}\right\rfloor
      \]
      if such a polynomial exists, or
    \item \fail{} otherwise.
  \end{enumerate}

  In particular, for any fixed multiplicity $s$, the algorithm achieves the above optimal unique-decoding radius whenever
  \[
    (n - 1) + (k-1) - 2\sqrt{(n-1)(k-1)(1+1/s)} > 0.
  \]
\end{theorem}

\begin{proof}
  By Theorem~\ref{thm:list_dec_rl__correctness}, Algorithm~\ref{alg:list_dec_rl} (and hence Step~\ref{step:unique_dec_rl__run_list_dec} of Algorithm~\ref{alg:unique_dec_rl}) returns a list whenever
  \[
    \tau = \left\lfloor\frac{n - k}{2}\right\rfloor
      < (n-1) - \sqrt{(n-1)(k-1)}.
  \]

  Since $\tau \le \frac{n-k}{2}$, it suffices that $\frac{n-k}{2} < (n-1) - \sqrt{(n-1)(k-1)}$, which is equivalent to
  \[
  \begin{split}
    & n - k < 2n - 2 - 2\sqrt{(n-1)(k-1)} \\
    \quad\Longleftrightarrow\quad
    & (n-1) + (k-1) > 2\sqrt{(n-1)(k-1)} \\
    \quad\Longleftrightarrow\quad
    & (\sqrt{n-1} - \sqrt{k-1})^2 > 0.
  \end{split}
  \]

  This is guaranteed by the assumption $k < n$. Hence, under this assumption, Algorithm~\ref{alg:list_dec_rl} does not output \fail{}, and the returned list $\mL$ consists precisely of all polynomials $f$ such that
  \[
    \rd\bigl(\enc(f), \br\bigr) \le \left\lfloor\frac{n-k}{2}\right\rfloor.
  \]

  Since $\mC_{\RL}(\bGa, \bv, k, \delta)$ is MDS, its minimum distance is $d = n-k+1$, and thus the unique-decoding radius is
  $\left\lfloor\frac{d-1}{2}\right\rfloor = \left\lfloor\frac{n-k}{2}\right\rfloor$.

  Consequently, at most one codeword of $\mC_{\RL}$ lies in the Hamming ball of $\br$ with radius $\left\lfloor\frac{n-k}{2}\right\rfloor$, and hence $\lvert \mL \rvert \le 1$.

  Finally, Algorithm~\ref{alg:unique_dec_rl} outputs the unique $f$ when $\lvert \mL \rvert = 1$, and outputs \fail{} when $\lvert \mL \rvert = 0$, which proves the claim.

  For fixed $s$, Theorem~\ref{thm:list_dec_rl__correctness} guarantees successful list decoding whenever
  \[
    \tau = \left\lfloor\frac{n-k}{2}\right\rfloor < (n-1) - \lfloor \sqrt{(n-1)(k-1)(1+1/s)}\rfloor.
  \]

  This inequality holds under the condition $(n - 1) + (k-1) - 2\sqrt{(n-1)(k-1)(1+1/s)} > 0$, completing the proof.
\end{proof}

\subsection{Decoding Roth--Lempel Codes with AMD Codes}

As in the case of TGRS codes, AMD codes can be employed to select the correct polynomial from the output list of the list decoding algorithm for Roth--Lempel codes. Since the construction and analysis are analogous to those in Section~\ref{sec:decoding_tgrs_amd}, we omit the full details and only state the corresponding result.

\begin{theorem}
  \label{thm:decoding_rl_amd}
  Let $\mC_{\RL}(\bGa, \bv, k+2b, \delta)$ be a Roth--Lempel code. Replace $\mC_{\TGRS}(\bGa, \bv, k+2b, \mT)$ by $\mC_{\RL}(\bGa, \bv, k+2b, \delta)$ in Algorithm~\ref{alg:encoding_tgrs_amd} and Algorithm~\ref{alg:decoding_tgrs_amd} to obtain, respectively, an encoding and a decoding procedure for Roth--Lempel codes augmented with AMD codes.

  Suppose a message $\m \in \Fqk$ is encoded by this encoding algorithm into a codeword $\bc \in \mC_{\RL}(\bGa, \bv, k+2b, \delta)$. Let $\br$ be a received word and let $\tau$ be a decoding radius such that
  \[
    \tau < (n-1) - \sqrt{(n-1)(k+2b-1)}.
  \]

  In particular, for any fixed multiplicity $s$, the algorithm succeeds whenever
  \[
    \tau < (n-1) - \lfloor \sqrt{(n-1)(k+2b-1)(1+1/s)} \rfloor.
  \]

  If $\rd(\br,\bc) \le \tau$, then the decoding algorithm outputs $\m$ with probability at least
  \[
    1 - \frac{\lceil k/b \rceil + 2}{q^b}\!\left(\sqrt{\frac{(n - 1)s(s+1)}{k+2b-1}} - 1\right).
  \]

  Counting operations over $\Fq$ as basic operations, the encoding algorithm runs in time $O(n \log^2 n)$, and the decoding algorithm runs in time
  \[
    O\!\left(\Bigl(s\frac{n}{k+2b}\Bigr)^{O(1)} n\log^2 n\log\log n\right).
  \]
\end{theorem}

\begin{proof}
  The correctness of the decoding algorithm follows from
  Theorem~\ref{thm:list_dec_rl__correctness} together with the same argument as
  in the proof of Theorem~\ref{thm:decoding_tgrs_amd__correctness}. In particular,
  under the stated condition on~$\tau$, the underlying list decoder never
  outputs \fail{}, and the AMD verification step then identifies the transmitted
  message with the claimed probability via a union bound over all nonzero
  manipulations.

  For the complexity, recall that encoding a Roth--Lempel codeword can be done
  in $O(n\log^2 n)$ time via fast multipoint evaluation. The additional AMD-tag
  computation costs $O(k\log b)$ operations over $\Fq$ by
  Lemma~\ref{lem:amd_complexity}, which is asymptotically dominated by
  $O(n\log^2 n)$ since $k \le n$ and $b \le n$. Hence the overall encoding
  complexity is $O(n\log^2 n)$.

  For the decoding procedure, invoking the list decoder for
  $\mC_{\RL}(\bGa, \bv, k+2b, \delta)$ costs
  \[
    O\!\left(\Bigl(s\frac{n}{k+2b}\Bigr)^{O(1)} n\log^2 n\log\log n\right)
  \]
  by Theorem~\ref{thm:list_dec_rl__complexity}. The AMD verification checks at
  most $\sqrt{(n-1)s(s+1)/(k+2b-1)} = O(s\sqrt{n/(k+2b)})$ candidates and, by Lemma~\ref{lem:amd_complexity},
  contributes
  \[
    O\!\left(s\sqrt{\frac{n}{k+2b}} \cdot k \log b\right)
  \]
  operations over $\Fq$, which is asymptotically dominated by the bound above.
  Thus the stated decoding complexity follows.
\end{proof}

We present a numerical example that illustrates the AMD-assisted decoding procedure for Roth--Lempel codes. We choose a received word with two candidates within the decoding radius and show how AMD identifies the intended message.

\begin{example}
  \label{exm:rl_with_amd}
  Consider the Roth--Lempel code $\mC_{\RL}(\{0,1,\ldots,22\}, \boldsymbol{1}, 6, 4)$ over $\F_{23}$. 
  Suppose the message $(18, 16, 9, 16)$ is encoded into the codeword
  \[
  \begin{split}
    (&18, 13, 7, 4, 8, 0, 7, 10, 13, 20, 12, \\
     &16, 13, 7, 2, 2, 11, 10, 3, 17, 10, 9, 18, 0).
  \end{split}
  \]

  Assume the transmission is affected by the error vector of Hamming weight $\lceil (n-1) - \sqrt{(n-1)(k-1)}\rceil - 1 = 12$:
  \[
  \begin{split}
    (&4, 21, 13, 0, 0, 0, 0, 0, 0, 0, 8, 22,\\
     &0, 0, 15, 8, 0, 9, 14, 10, 22, 0, 10, 0).
  \end{split}
  \]

  Running the list decoding algorithm on the received word yields two candidate
  AMD-augmented messages:
  \[
    (22, 2, 2, 20, 7, 4), \qquad (18, 16, 9, 16, 0, 0).
  \]

  The first candidate fails the AMD verification, while the second one passes and
  corresponds to the originally transmitted message.
\end{example}

\section{Conclusion}

\label{sec:conclusion}

In this work, we developed Guruswami--Sudan-based decoding algorithms for twisted GRS codes and Roth--Lempel codes.

For TGRS codes, we viewed them as subcodes of suitable GRS codes and exploited this relation to obtain an efficient list decoder that corrects up to $\tau < n - \sqrt{n(k'-1)}$ errors. Building on this result, we further derived a unique decoder for MDS TGRS codes that, for fixed rate, fixed constant multiplicity $s$ and $k' < \frac{(n + k)^2}{4n(1+1/s)} + 1$, runs in near-linear time and supports up to $\ell = O(n\log^2 n \log\log n)$ twists.

For Roth--Lempel codes, we gave what appears to be the first explicit decoding algorithms, achieving list decoding radius $\tau < (n-1) - \sqrt{(n-1)(k-1)}$ and unique decoding up to half the minimum distance when $k < n$.

In both settings we combined list decoding with optimal systematic AMD codes to obtain probabilistic decoders that, with high probability, recover the unique transmitted message beyond the classical MDS bound.

\ifCLASSOPTIONcaptionsoff
  \newpage
\fi

\bibliographystyle{IEEEtran}
\bibliography{references.bib}

\begin{IEEEbiographynophoto}{Runtian Zhu}
received the B.E. degree in computer science from the School of Computer Science and Technology, East China Normal University, Shanghai, China, in 2023. He is currently pursuing the Ph.D. degree with the College of Computer Science and Artificial Intelligence, Fudan University, Shanghai, China. His research interests include coding theory and its applications.
\end{IEEEbiographynophoto}

\begin{IEEEbiographynophoto}{Lingfei Jin}
(Member, IEEE) received the Ph.D. degree from Nanyang Technological University, Singapore, in 2013. Currently, she is a Full Professor with the College of Computer Science and Artificial Intelligence, Fudan University, Shanghai, China. Her research interests include classical coding theory and quantum coding theory.
\end{IEEEbiographynophoto}

\end{document}